\documentclass[conference,letterpaper]{IEEEtran}
\usepackage[letterpaper, left=0.8in, right=0.8in, bottom=0.8in, top=0.8in]{geometry}

\IEEEoverridecommandlockouts

\usepackage{graphics}
\usepackage{epsfig}
\usepackage{amsmath}
\usepackage{amssymb}
\usepackage{xcolor}
\usepackage{multirow}

\usepackage{amsthm}
\theoremstyle{definition}

\newtheorem{theorem}{\normalfont\bfseries Theorem}

\newtheorem{proposition}{\normalfont\bfseries Proposition}
\newtheorem{definition}{\normalfont\bfseries Definition}

\newtheorem{corollary}{\normalfont\bfseries Corollary}

\newtheorem{remark}{\normalfont\bfseries Remark}

\usepackage{hyperref}
\usepackage{graphicx}		
\usepackage{wrapfig}
\usepackage[format=plain,font=footnotesize,labelfont=bf,labelsep=period]{caption}
\usepackage[export]{adjustbox}
\usepackage{bm}

\usepackage{amsmath}
\usepackage{amssymb}
\usepackage{mathtools}
\usepackage{algorithm}
\usepackage{algpseudocode}

\usepackage{pdfsync}
\usepackage[normalem]{ulem}
\usepackage{paralist}	
\usepackage[space]{grffile} 

\usepackage{color}
\newenvironment{proof}{\textit{Proof.}}{\hfill$\square$}

\newcommand{\R}{\mathbb{R}}

\definecolor{darkblue}{RGB}{0,0,102}
\definecolor{lightblue}{RGB}{77,77,148}

\definecolor{gold}{RGB}{234, 170, 0}
\definecolor{metallic_gold}{RGB}{139, 111, 78}

\newcommand{\mb}[1]{\mathbf{ #1 }}

\DeclareMathOperator*{\argmin}{argmin}

\newcommand{\lmat }{\begin{bmatrix}}
\newcommand{\rmat}{\end{bmatrix}}

\newcommand{\etek}{CBF-compliant controller}

\title{\LARGE \bf End-to-End Imitation Learning with Safety Guarantees using Control Barrier Functions}

\author{Ryan K. Cosner, Yisong Yue, Aaron D. Ames
\thanks{This research is supported by BP and Aerovironment.}
\thanks{ R.K. Cosner and A.D. Ames are with the Department of Mechanical and Civil Engineering and Y. Yue is with the Department of Computing and Mathematical Sciences, California Institute of Technology, Pasadena, CA, 91125, USA,  {\tt\small \{rkcosner, yyue, ames\}@caltech.edu}. Y. Yue is also affiliated with Argo AI, Pittsgurgh, PA.  }%
}

\IEEEaftertitletext{\vspace{-1em}}

\begin{document}

\maketitle

\begin{abstract}
Imitation learning (IL) is a learning paradigm 
which can be used to synthesize controllers for complex systems that
mimic behavior demonstrated by an expert (user or control algorithm). Despite their popularity, IL methods generally lack guarantees of safety, which limits their utility for complex safety-critical systems. In this work we consider safety, formulated as set-invariance, and the associated formal guarantees endowed by Control Barrier Functions (CBFs). We develop conditions under which robustly-safe expert controllers, utilizing CBFs, can be used to learn end-to-end controllers (which we refer to as \textit{CBF-Compliant} controllers) that have safety guarantees. These guarantees are presented from the perspective of input-to-state safety (ISSf) which considers safety in the context of 
disturbances, wherein it is shown that IL using  robustly safe expert demonstrations results in ISSf with the 
disturbance directly related to properties of the learning problem.  We demonstrate these safety guarantees in simulated vision-based end-to-end control of an inverted pendulum and a car driving on a track. 

\end{abstract}


\section{Introduction}

The use of learning in conjunction with control has become  increasingly popular---especially in the context of autonomous systems and robotics where system properties, e.g. stability and safety, must generalize to the real world. 
Of particular interest in this paper, Imitation Learning (IL) trains a policy to mimic behavior demonstrated by an expert \cite{hussein2017imitation}. IL is a paradigm 
that has shown significant success in video-games \cite{ross2011reduction}, humanoid  robotics \cite{schaal1999imitation}, and autonomous vehicles \cite{codevilla2018end, pan2017agile}.  The goal of this paper is to extend the theoretic underpinnings of safety to IL applications.

Safety, framed as the forward-invariance of a \textit{safe set}, has become a dominant definition in control theory. Several methods have been introduced which provide guarantees of safety including model predictive control (MPC) \cite{wabersich_linear_2018}, optimal reachability-based methods \cite{bansal2017hamilton}, and control barrier functions (CBFs) \cite{ames_control_2017, borrmann2015control, nguyen2016exponential}. 
Control barrier function methods use a Lyapunov-like condition to guarantee safety and present advantages over other methods in that they provide guarantees for general 
continuous-time control-affine nonlinear systems and can be implemented efficiently as convex optimization programs. The robustness of CBF-based safety methods has also been  studied in the context of state uncertainty \cite{dean_guaranteeing_2020}, dynamics uncertainty \cite{kolathaya2018input}, and reduced-order models \cite{molnar_model-free_2021}.

With the goal of deploying 
machine learning to real-world systems where safety is critical, 
there has been a large body of recent work exploring the connection between machine learning and safety-critical control. 
For example, MPC controllers have incorporated learning to improve performance while satisfying state constraints in the presence of uncertainty \cite{rosolia2017autonomous} and model predictive safety certificates have been used to ensure the safety of learned controllers \cite{wabersich2021predictive}.  Learning has also been used in optimal control to ensure safety in the context of reinforcement learning \cite{fisac2019bridging}.  Finally, several CBF methods learn the uncertainty in either the full-order dynamics \cite{emam2021safe, castaneda2021pointwise} or the projection of the uncertainty to the CBF \cite{taylor_learning_nodate, csomay2021episodic} to better enforce safety. 

\begin{figure}[t]
    \centering
    \includegraphics[width=0.87\linewidth]{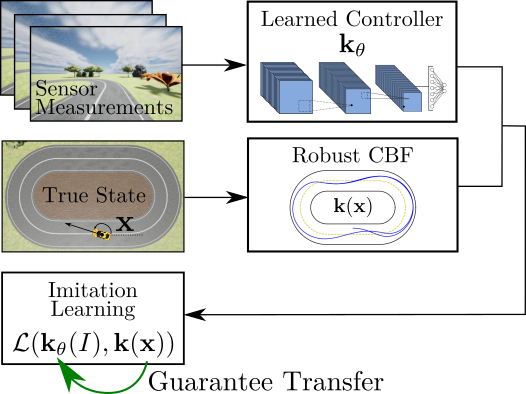}
    \caption{The robust control barrier function-based controller $\mb{k}$ is used as the expert controller in an imitation learning problem to generate the end-to-end controller $\mb{k}_\theta$. By ensuring that $\mb{k}$ has robust safety guarantees, we ensure that $\mb{k}_\theta$ is also guaranteed to keep the system safe.}
    \label{fig:high_level_fig}
    \vspace{-0.85cm}
\end{figure}



Specifically in the context of IL, safety of learned controllers has been studied from several perspectives. In the simplest setting, the controller can be disengaged whenever it is deemed unsafe by a human operator \cite{pan2017agile}. By defining safety as the deviation from the expert controller, several methods such as SafeDAgger \cite{brown2018risk} and EnsembleDAgger \cite{zhang2017query} can achieve this form of safety by returning control of the system to the expert whenever an uncertainty threshold is reached. This perspective on safety is different from set-invariance as small deviations from a safe controller may cause the system to leave the desired safe set. 
Alternatively, forward-invariance for systems with IL-based controllers for discrete-time systems has been studied in the context Lyapunov stability 
\cite{yin2021imitation} and MPC \cite{hertneck2018learning}. The success of IL-based controllers learned from CBF-based demonstrations has been shown in simulation without guarantees  \cite{yaghoubi2020training}.
Yet theoretical connections between IL and CBF-based expert controllers have not been established.

This paper presents an end-to-end IL method that endows the learned controller with formal guarantees of safety, as illustrated in Fig. \ref{fig:high_level_fig}.  
The central idea is that if the expert controller has robust CBF-based safety guarantees and the proper sampling scheme is employed then it is possible to obtain safety guarantees for the IL-synthesized controllers, we call these learned controllers \textbf{CBF-compliant}.
These results are framed in the context of robust safety, and specifically input-to-state safety (ISSf), wherein we established connections between the learning methods, sampling density, network smoothness, and learning accuracy.  This manifests in the main result of the paper, which gives a formal guarantee of safety (in an ISSf context) where properties of the learning problem directly affect the corresponding robust safe set.  This result is demonstrated in two simulations: an inverted pendulum and a vehicle driving on a track.  In both cases, safe vision-based control is demonstrated using our end-to-end safe imitation learning methodology, where end-to-end refers to direct perception-to-control.


\section{Problem Setup}
\label{sec:problem}

Throughout the rest of this work we consider nonlinear systems with a control-affine structure:
\begin{align}\label{eq:openloop}
    \dot{\mb{x}} & = \mb{f}(\mb{x}) + \mb{g}(\mb{x}) \mb{u},
\end{align}
\noindent where $\mb{x} \in \R^n$ and $\mb{u} \in \R^m $ are the states and inputs respectively and also where the drift dynamics $\mb{f}: \R^n \to \R^n$ and input matrix $\mb{g}: \R^n \to \R^{n\times m} $ are locally Lipschitz continuous on $\R^n$. Given a continuously differentiable state-feedback controller $\mb{k}: \R^n \to \R^m $ that has access to full state information, the closed loop dynamics are given by:
\begin{align}\label{eq:closed_loop}
    \dot{\mb{x}} = \mb{f}_\textrm{cl}(\mb{x}) \triangleq \mb{f}(\mb{x}) + \mb{g}(\mb{x})\mb{k}(\mb{x}).  
\end{align}

The assumption of continuous differentiability of $\mb{f},$ $ \mb{g}, $ and $\mb{k}$ implies that $\mb{f}_\textrm{cl}$ is locally Lipschitz continuous. Thus for any initial condition $\mb{x}_0 \triangleq \mb{x}(0) \in \R^n $ there exists a time interval $I(\mb{x}_0) = [0, t_\textrm{max}) $ such that $\mb{x}(t)$ is the unique solution to \eqref{eq:closed_loop} on $I(\mb{x}_0)$ \cite{perko2013differential}. We further assume that the closed loop system \eqref{eq:closed_loop} is \textit{forward complete}, i.e.  $T_\textrm{max} = \infty $. 

In the remainder of this section we provide preliminaries on imitation learning and control barrier functions.

\subsection{Imitation Learning}

Imitation learning (IL) is a common learning framework in which a mapping between observations and actions is trained using expert demonstrations. Common methodologies in IL include behavioral cloning (a form of supervised learning), DAgger \cite{ross2011reduction} which repeatedly collects additional state-action pairs, and inverse reinforcement learning (IRL) \cite{ziebart2008maximum} which learns a cost function such that the action or action sequence with minimal cost agrees with the expert demonstrations. 
We will present our method in the context of behavioral cloning, but note that our method is not specific to this form of IL and can be generalized to provide safety guarantees for DAgger and IRL since the theory developed in this work depends on the learned controller itself and not the learning framework used to produce it.

For end-to-end IL we model sensor measurements as: 
\begin{align}
    \mb{y} = \mb{c}(\mb{x}) \label{eq:camera}
\end{align}
\noindent where $\mb{y} \in \R^k $ represents the system observations and $\mb{c}: \R^n \to \R^k $ represents the system's sensors which we assume to be locally Lipschitz as in \cite{dean_guaranteeing_2020}. In the context of computer vision $\mb{y}$ may be a vector representation of image data and $\mb{c} $ may be the camera sensor which maps from state to image. 

In order to train an end-to-end controller, we collect an exogenous dataset of observation-input pairs using the expert controller $\mb{k}:\R^n \to \R^m$: 
\begin{align}
    \mathcal{D} =  \{D_i \}^N_{i=1},  &&  D_i = (\mb{c}(\mb{x}_i), \mb{k}(\mb{x}_i) ) \in \R^k \times \R^m 
\end{align}
for $N \in \mathbb{N}$ samples\footnote{As in \cite{pan2017agile}, we assume the expert controller has access to the true state. }.  Given a nonlinear function class $\mathcal{H}: \R^k \to \R^m $ and a loss function $\mathcal{L}: \R^m\times\R^m \to \R$, the learning problem can be expressed as optimizing the parameters $\theta$ of the function $\mb{k}_\theta \in \mathcal{H}$ via empirical risk minimization: 
\begin{align}
    \min_{\mb{k}_\theta \in \mathcal{H}} \frac{1}{N} \sum_{i=1}^N \mathcal{L}\left( \mb{k}_\theta(\mb{c}(\mb{x}_i)), \mb{k}(\mb{x}_i) \right) \label{eq:learning_opt_statement}. 
\end{align}

\begin{remark}
Behavioral cloning as in \eqref{eq:learning_opt_statement} suffers from compounding errors in the resulting trajectories \cite{ross2011reduction}. However, since our goal is  to transfer safety guarantees from the expert controller to the learned controller rather than to exactly mimic the expert behavior, we show that forward-invariance can be achieved despite compounding errors if the expert controller enforces robust forward-invariance as in Def. \ref{def:etek}.
\end{remark}

\subsection{Safety and Control Barrier Functions } \label{subsec:safety}
In this section we define safety as the forward-invariance of a \textit{safe set} $\mathcal{C} \subset \R^n $ and discuss control barrier functions (CBFs) as a tool for achieving safety for nonlinear systems. 
\begin{definition}[Forward Invariance and Safety]
A set $\mathcal{C}\subseteq \R^n $ is forward invariant if for every $\mb{x}_0 \in \mathcal{C}$ the solution to \eqref{eq:closed_loop}
 satisfies $\mb{x}(t) \in \mathcal{C} $ for all $ t \geq 0 $. The system \eqref{eq:closed_loop} is \textit{safe} with respect to the set $\mathcal{C}$ if $\mathcal{C}$ is forward invariant. 
 \end{definition}
 
Consider the set $\mathcal{C}$ defined as the 0-superlevel set of some continuously differentiable function $h:\R^n \to \R$: 
 \begin{align}
    \mathcal{C} &\triangleq \{ \mb{x} \in \R^n ~|~ h(\mb{x}) \geq 0\} \\
    \textrm{Int}(\mathcal{C}) &\triangleq\{ \mb{x} \in \R^n ~|~ h(\mb{x}) > 0\} \\
    \partial \mathcal{C} &\triangleq \{ \mb{x} \in \R^n ~|~ h(\mb{x}) = 0\} 
 \end{align}
 \noindent These functions $h$ are used in the following definition of CBFs:
 \begin{definition}[Control Barrier Function (CBF),  \cite{ames_control_2017}]
 Let $\mathcal{C}\subseteq \R^n $ be the 0-superlevel set of a continuously differentiable function $h: \R^n \to \R$ with 0 a regular value\footnote{$0$ is a regular value of $h: \R^n \to \R$ if $h(\mb{x}) = 0  \implies \frac{\partial h }{\partial \mb{x} }(\mb{x}) \neq 0 $}. The function $h$ is a \textit{Control Barrier Function (CBF)} for \eqref{eq:openloop} on $\mathcal{C}$ if there exists a locally Lipschitz extended class $\mathcal{K}$ infinity function\footnote{ 
 A continuous function $\alpha:\R_{\geq 0 } \to \R_{\geq 0}$ is \textit{class } $\mathcal{K}_\infty$  ($\alpha \in \mathcal{K}_\infty$) if $\alpha(0) = 0 $, $\alpha $ is strictly monotonically increasing, and $\lim_{c \to \infty } \alpha(c ) = \infty $. 
 We say that a continuous function $\alpha: \R \to \R $ is \textit{extended class } $\mathcal{K}_\infty^e$ ($\alpha \in \mathcal{K}_{\infty}^e $) if $\alpha(0) = 0$, $\alpha $ is strictly monotonically increasing, $\lim_{c \to \infty } \alpha (c) = \infty $ and $\lim_{c \to - \infty } \alpha(c) = - \infty $.} 
 $\alpha \in \mathcal{K}_{\infty}^e$ such that for all $\mb{x} \in \mathcal{C}$: 
 \begin{align}\label{eq:cbf_constraint}
     \sup_{\mb{u}\in\R^m} \overbrace{ \underbrace{\frac{\partial h}{\partial \mb{x}} (\mb{x}) \mb{f}(\mb{x})}_{L_\mb{f}h(\mb{x})} + \underbrace{\frac{\partial h}{\partial \mb{x}}(\mb{x}) \mb{g}(\mb{x})}_{L_\mb{g}h(\mb{x})} \mb{u}}^{\dot{h}(\mb{x}, \mb{u})} \geq - \alpha (h(\mb{x})). 
 \end{align}
 \end{definition}
\noindent This safety condition \eqref{eq:cbf_constraint} constrains the time derivative of $h$, preventing it from being negative if $h(\mb{x}(0)) =  0$ and thus  rendering the 0-superlevel set $\mathcal{C}$ forward invariant.

This mathematical guarantee of safety is formalized as:  
 \begin{theorem}[\cite{ames_control_2017}]
Given a set $\mathcal{C}\subseteq \R^n $ defined as the 0-superlevel set of a continuously differentiable function $h: \R^n \to \R$, if $h$ is a CBF for \eqref{eq:openloop} on $\mathcal{C}$, then any locally Lipschitz continuous controller $\mb{k}:\R^n \to \R^m $, such that $\mb{k}(\mb{x}) = \mb{u}$ satisfies \eqref{eq:cbf_constraint}, renders system \eqref{eq:closed_loop} safe with respect to $\mathcal{C}$. \label{thm:cbfs}
 \end{theorem}
 
Beyond verifying the safety of closed-loop systems, CBFs can also be used as a tool for generating safe control inputs. One such controller that satisfies \eqref{eq:cbf_constraint}  is the CBF-QP:
\begin{align} \label{eq:CBF-QP} \tag{CBF-QP}
    \mb{k}_\textrm{cbf-qp} = \argmin_{\mb{u}\in \R^m } & \quad \frac{1}{2}\Vert \mb{u} - \mb{k}_\textrm{nom}(\mb{x})  \Vert^2\\
    & \quad \textrm{s.t. }  L_\mb{f}h(\mb{x}) + L_\mb{g}h(\mb{x}) \geq - \alpha (h(\mb{x})). \nonumber
\end{align}
The controller $\mb{k}_\mathrm{cbf-qp}: \R^n \to \R^m $ filters a nominal and potentially unsafe controller $\mb{k}_\textrm{nom}:\R^n \to \R^m$ and returns an input which satisfies the safety constraint \eqref{eq:cbf_constraint}. 

\section{Robust Safety}
\label{sec:robust}

In this section we discuss robust safety, which will allow us to capture the uncertainties generated by IL.  In particular, we begin by giving a summary of input-to-state safety (ISSf) followed by a discussion of the continuity properties of level sets of $h$ which, again, will allow us to reason about uncertainty as it is reflected in the safe set $\mathcal{C}$ given by $h$.

\subsection{Input-to-state Safety}

We begin by discussing robust safety and the Input-to-State Safety (ISSf) property \cite{kolathaya2018input} which extends CBF-based set-invariance to systems with disturbances. To do this we consider the system \eqref{eq:openloop} and introduce a matched, bounded, and potentially time-varying disturbance   $\mb{d}: \R_{\geq 0 } \to \R^m  $, 
\begin{align}
    \dot{\mb{x}}  = \mb{f}(\mb{x}) + \mb{g}(\mb{x})( \mb{k}(\mb{x}) + \mb{d}(t)) \label{eq:dyn_disturbed}.
\end{align}

Due to the disturbance, the safety guarantees established in Theorem \ref{thm:cbfs} no longer hold. To analyze the effect of the disturbance on safety we instead consider the expanded safe set $\mathcal{C}_\delta \supseteq \mathcal{C}$ for some $\delta \geq 0 $: 
\begin{align}
    \mathcal{C}_\delta &  \triangleq \{ \mb{x} \in \R^n ~|~ h(\mb{x}) \geq -\gamma (\delta)  \} \label{eq:issf_safe_set}\\
    \textrm{Int}(\mathcal{C}_\delta) & \triangleq  \{ \mb{x} \in \R^n ~|~ h (\mb{x} ) > - \gamma(\delta)  \} \\
    \partial \mathcal{C}_\delta    & \triangleq \{ \mb{x} \in \R^n ~|~ h (\mb{x}) = - \gamma(\delta) \}.
\end{align}
where $\gamma \in \mathcal{K}_\infty$. We note that the original safe set is recovered when $\delta = 0 $, i.e.  $\mathcal{C}_0 = \mathcal{C}$.

We can now define Input-to-State-Safety (ISSf) as the forward invariance of the set $\mathcal{C}_\delta$ in the presence of disturbances. 
\begin{definition} [Input-to-State Safety (ISSf) \cite{alan_safe_2022}]
Let $\mathcal{C}\subset \R^n $ be the 0-superlevel set of a continuously differentiable function $h: \R^n \to \R $. The system \eqref{eq:dyn_disturbed} is \textit{input-to-state safe} (ISSf) if there exists $\gamma \in \mathcal{K}_\infty$ and $\delta \geq 0 $ such that $\forall  \mb{d} $ where $\Vert \mb{d} \Vert_\infty \leq \delta  $, the set $\mathcal{C}_\delta$ defined by \eqref{eq:issf_safe_set} is safe. In this case we refer to the set $\mathcal{C}$ as an \textit{input-to-state safe set} (ISSf set). 
\end{definition}
\noindent ISSf is the safety analog to the more common Input-to-State Stability (ISS) property of Lyapunov stable systems \cite{sontag2008input}. 


\subsection{Continuity Properties}
Before proving the safety-transfer that occurs between the expert and learned controllers in \eqref{eq:learning_opt_statement}, we must first establish the upper semi-continuity (USC) of the level sets of $h$ which will allow us to reason about the expanded safe set $\mathcal{C}_\delta$.

To discuss non-zero level sets, we define the $c$-level set of a function $h$ using the preimage $h^{-1} : \R \rightsquigarrow \mathcal{P}(\R^n)$, 
\begin{align}
    h^{-1}(c) = \{ \mb{x} \in \R^n ~|~ h(\mb{x}) = c \} \label{eq:levelset}
\end{align}
where $c \in \R$ and $\mathcal{P}(\R^n)$ denotes the power set of $\R^n$. 

One definition used to describe the continuity properties of point-to-set maps is USC:  
\begin{definition}[Upper Semi-Continuity (USC) \cite{aubin2009set}]
A set valued function map $h^{-1}: \R \rightsquigarrow \mathcal{P}(\R^n) $ is \textit{upper semi-continuous} at $c \in \R $ if and only if for any $\epsilon >0$, there exists $\eta > 0 $ such that $ c' \in B_\eta (c) \implies h^{-1}(c') \subset h^{-1}(c) \oplus B_\epsilon   $. 
\end{definition}
\noindent where $B_\eta $ is the Euclidean ball of radius $\eta$. 

It was established in \cite[Proposition 6]{konda2019characterizing} that, under common assumptions for CBFs, the level sets $h^{-1}$  are USC.
\begin{proposition}[\cite{konda2019characterizing}]
Let $h^{-1}: \R \rightsquigarrow \mathcal{P}(\R^n) $  be the preimage \eqref{eq:levelset} representing the $c$-level set of some continuously differentiable function $h: \R^n \to \R$. If 0 is a regular value of $h$ and $\Lambda := \{ \mb{x} \in \R^n ~|~ -\delta \leq h(\mb{x}) \leq \delta \} $ is compact for all $\delta \geq 0 $, then $h^{-1} $ is upper semi-continuous at 0.   \label{prop:usc}
\end{proposition}



This proposition relates the regularity of $h$ to the upper semi-continuity of its level sets. In essence, the regularity of $h$ at $0$ ensures that small changes in the value defining the level set has a small effect on the level set itself. 

\section{Main Result} \label{sec:main_result}
In this section we present our main result relating the supervised training of end-to-end controllers to intput-to-state safety. To render the following closed loop system safe: 
\begin{align}
    \dot{\mb{x}} & = \mb{f}(\mb{x}) + \mb{g}(\mb{x})\mb{k}_\theta(\mb{c}(\mb{x})) \label{eq:dyn_ete}.
\end{align}

For the behavioral cloning problem \eqref{eq:learning_opt_statement} we require an expert controller that provides robustness to matched disturbances and state uncertainty. 
 For this we choose the Tunable Robust Optimization Program \ref{eq:trop} controller \cite{cosner_safety-aware_2021}. 

\begin{definition}[Tunable Robust Optimization Program (TR-OP) Controller ] \label{eq:expert_controller}
\begin{align}
    \mb{k}_T({\mb{x}}) = \argmin_{\mb{v} \in \R^m } & \Vert \mb{v} - \mb{k}_\textrm{nom}({\mb{x}}) \Vert^2 \tag{TR-OP} \label{eq:trop}\\
    \textrm{s.t. } & L_\mb{f}h({\mb{x}}) + L_\mb{g}h({\mb{x}})\mb{v} \nonumber \noindent\\ 
    &  - \varphi \Vert L_\mb{g}h({\mb{x}})\Vert^2 - a -b \Vert \mb{v} \Vert \geq - \alpha ( h ({\mb{x}})). \nonumber
\end{align}
with parameters $\varphi , a, b \in \R_{\geq 0 }$ and $\alpha \in \mathcal{K}_\infty^e$.
\end{definition}
\noindent  $\varphi, a, $ and $b$ effect the robustness of the closed loop system with respect to matched disturbances and state uncertainty. 

\begin{definition}[CBF-Compliancy] \label{def:etek}
The learned controller $\mb{k}_\theta: \R^k \to \R^m $ is a \textit{CBF-Compliant}~ for some $h: \R^n \to \R$ with measurement function $\mb{c}: \R^n \to \R^k $ if 
\begin{align}
    \min_{\mb{x} \in \mathcal{D} } \Vert \mb{x}_1 - \mb{x}  \Vert  &\leq r_1,  \label{eq:thm_required_sampling}\\
    \Vert \mb{k}_T(\mb{x}_2) - \mb{k}_\theta(\mb{c}(\mb{x}_2))\Vert  &\leq M_e, \label{eq:error_bound} \\
    \Vert \mb{k}_\theta(\mb{c}(\mb{x}_3)) - \mb{k}_\theta(\mb{c}(\mb{x}_4)) \Vert   &\leq \mathfrak{L}_{\mb{k}_\theta\circ \mb{c}} \Vert \mb{x}_3 - \mb{x}_4 \Vert,  \quad   \label{eq:controller_lip_bound}
\end{align}
for all $\mb{x}_1 \in \partial \mathcal{C}$, $(\mb{x}_2, \mb{k}_T(\mb{x}_2)) \in \mathcal{D}$, and $\mb{x}_3, \mb{x}_4 \in \partial \mathcal{C} \oplus \overline{B}_{r_2}$ where $r_1, r_2\in \R_{>0} $, $ \mathfrak{L}_{\mb{k}_\theta}, M_e\in \R_{\geq 0 }$, and $\mb{k}_T: \R^n \to \R^m $ is a \ref{eq:trop} controller for the function $h: \R^n \to \R$ with parameters $\varphi, a,b\in \R_{\geq 0 }$ and $\alpha \in \mathcal{K}_\infty^e$. 
\end{definition}

\noindent Note that $\oplus $ indicates the Minkowski sum.

\begin{remark}
Intuitively, any trajectory $\mb{x}(t)$ that may leave an expanded safe set $\mathcal{C}_\delta\supset \mathcal{C}$ must pass through the boundary $\partial \mathcal{C}_\delta$. Forward invariance of $\mathcal{C}_\delta$ and upper semi-continuity of $h$ ensure that the trajectory remains sufficiently close to a point in $\partial \mathcal{C}$ for which we have expert data, thus preventing the cascading failure typical of behavioral cloning.
\end{remark}

Next, we use the USC of $h$ to relate the existence of a \etek~to the ISSf of system \eqref{eq:dyn_ete}. 
\begin{theorem}
Let $\mathcal{C}\subset \R^n $ be the 0-superlevel set of a some function $h:\R^n \to \R$ which satisfies the Proposition \ref{prop:usc}.

There exist $\underline{\varphi},\underline{a}, \underline{b}\in\R_{\geq 0}$ such that if
\begin{itemize}
    \item $\mb{k}_\theta: \R^k \to \R^m$ is a \etek~on $\mathcal{C}$ for parameters $\varphi\geq \underline{\varphi}, a\geq\underline{a},$ $b\geq\underline{b}$, $\alpha \in \mathcal{K}_\infty^e$ with constants $r_1, r_2 > 0$ 
    \item and  $L_\mb{f}h, L_\mb{g}h, \Vert L_\mb{g}h \Vert^2$, and $\alpha \circ h$ be Lipschitz continuous on $\partial\mathcal{C} \oplus \overline{B}_{r_2}$,
\end{itemize}
then the closed loop system \eqref{eq:dyn_ete} is ISSf with respect to $\mathcal{C}$ and safe with respect to
\begin{align}
    \mathcal{C}_\delta = \left\{ \mb{x} \in \R^n \bigg | h(\mb{x}) \geq \alpha^{-1} \left( \frac{-1}{2\varphi} (\mathfrak{L}_{\mb{k}_\theta \circ \mb{c}}r_3 + M_\mb{e})^2 \right)  \right\}. \label{eq:theorem_issf_set}
\end{align}
where $r_3 \triangleq r_2 + r_1 $.
 \label{thm:main}
\end{theorem}

\begin{figure*}[t]
    \centering
    \includegraphics[width=\linewidth]{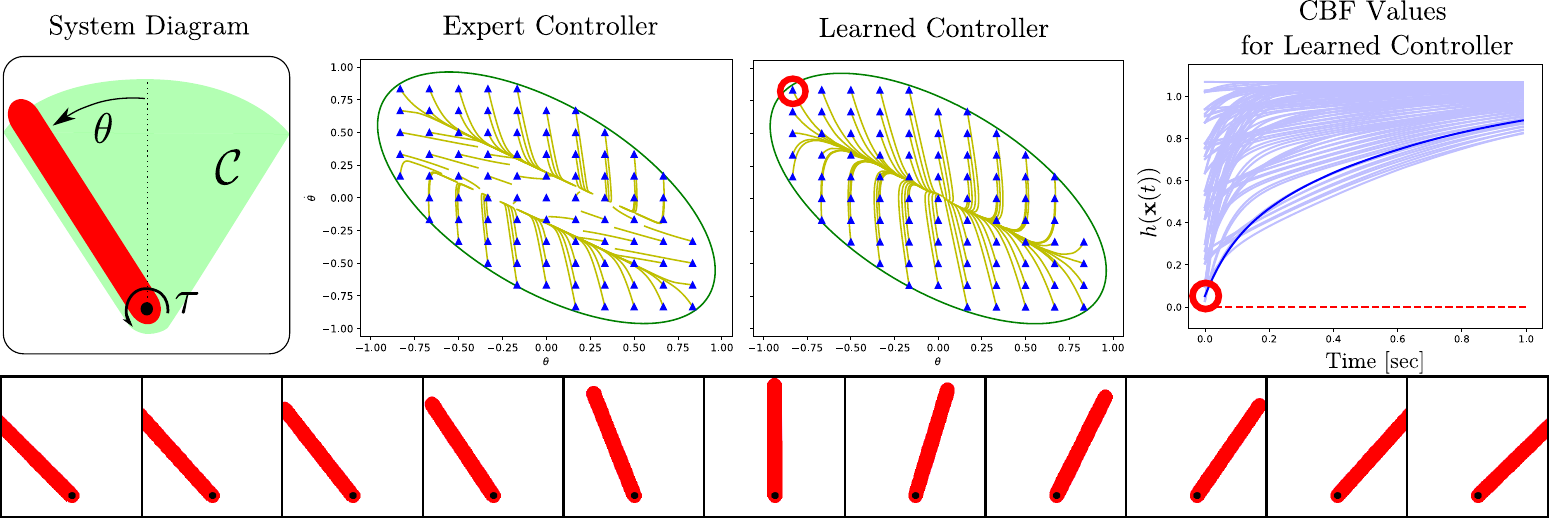}
    \caption{Results for the inverted pendulum. 
    \textbf{Left: } A diagram of the system where the green represents the safe set $\mathcal{C}$. 
    \textbf{Center Left:} One second long trajectories generated by the expert controller $\mb{k}_T$ are shown in yellow and plotted for several initial conditions represented by blue triangles. The green ellipse marks the boundary $\partial \mathcal{C}$. 
    \textbf{Center Right: } 
    One second long trajectories generated by the learned controller $\mb{k}_\theta$ are shown in yellow and plotted for several initial conditions represented by blue triangles. in The green ellipse marks the boundary $\partial \mathcal{C}$. 
    \textbf{Right: }
    CBF  values $h (\mb{x}(t))$ achieved by the learned controller. Note all are greater than zero indicating safety of the system. The darker blue trajectory begins at the initial condition marked by the red circle in the plot of the learned controller. \textbf{Bottom: } Images spanning the safe set $\mathcal{C}$ that are used by $\mb{k}_\theta$ for end-to-end control of the system. 
   }
    \label{fig:ip_results}
    \vspace{-0.5cm}
\end{figure*}
\begin{proof}
Consider some state $\mb{x}_3 \in \partial \mathcal{C} \oplus B_{r_2}$. 
Given $\mb{x}_3$ there must exist some $\mb{x}_2 \in \partial \mathcal{C}$ such that $\Vert\mb{x}_2 - \mb{x}_3  \Vert\leq r_2$. Additionally, since $\mb{k}_\theta$ is a \etek~there must be some $\mb{x}_1 \in \mathcal{D}$  such that $\Vert \mb{x}_1 - \mb{x}_2 \Vert \leq r_1$ and so $\Vert \mb{x}_1 - \mb{x}_3\Vert \leq r_3$ by the triangle inequality.

The function $h$ satisfies Proposition \ref{prop:usc} by assumption so by the definition of upper semi-continuity it follows that:  
\begin{align}
    \exists \eta>0 \textrm{ s.t. }  c \in B_\eta \implies h^{-1}(c) \subset h^{-1}(0) \oplus B_{r_2} . 
\end{align} 

Thus we can choose $\underline{\varphi} >0 $ large enough such that for any  $\varphi \geq \underline{\varphi}$ we have that  $\partial \mathcal{C}_\delta\subset \partial \mathcal{C} \oplus B_{r_2}$ for $\mathcal{C}_\delta $ as in \eqref{eq:theorem_issf_set}.  Next we choose the remaining parameter bounds to be: 
\begin{align}
    a\geq \underline{a} & =  r_3(\mathfrak{L}_{L_\mb{f}h} + \mathfrak{L}_{\alpha\circ h } + \mathfrak{L}_{{\varphi}\Vert L_\mb{g}h \Vert^2}), \\
    b \geq \underline{b} & =  r_3  \mathfrak{L}_{L_\mb{g}h}.
\end{align}
where $\mathfrak{L}$ represents the Lipschitz constant of the subscripted function on $\partial \mathcal{C}\oplus B_{r_2}$. 

Using $\mb{k}_\theta$ we can bound the time derivative of the CBF at $\mb{x}_3\in \mathcal{C}\oplus B_{r_2} $ as:
\begin{align}
& \frac{d}{dt} h (\mb{x}_3, \mb{k}_\theta (\mb{c}(\mb{x}_3) ))  \\
& = L_\mb{g}h(\mb{x}_1)\mb{k}_T(\mb{x}_1)  + \frac{d}{dt}h(\mb{x}_3, \mb{k}_\theta(\mb{c}(\mb{x}_3)) ) - L_\mb{g}h(\mb{x}_1)\mb{k}_T(\mb{x}_1) \nonumber \\
& \geq L_\mb{f}h(\mb{x}_3) - L_\mb{f}h(\mb{x}_1) +  r_3 \mathfrak{L}_{L_\mb{f}h} \label{eq:pf_lipschitz_expansion} \\
& \quad + \alpha(h(\mb{x}_3)) - \alpha(h(\mb{x}_1)) +  r_3\mathfrak{L}_{\alpha \circ h } \nonumber\\
& \quad + \varphi\Vert L_\mb{g}h(\mb{x}_1) \Vert^2 - \varphi \Vert L_\mb{g}h(\mb{x}_3) \Vert^2 +  r_3 \mathfrak{L}_{\varphi \Vert L_\mb{g}h \Vert^2}\nonumber \\
& \quad + L_\mb{g}h(\mb{x}_3)\mb{k}_T(\mb{x}_1) - L_\mb{g}h(\mb{x}_1) \mb{k}_T(\mb{x}_1) +  r_3  \mathfrak{L}_{L_\mb{g}h} \Vert \mb{k}_T(\mb{x}_1) \Vert \nonumber\\
& \quad  - \alpha(h(\mb{x}_3))+ \varphi \Vert L_\mb{g}h(\mb{x}_3)\Vert^2 \nonumber \\
& \quad + L_\mb{g}h(\mb{x}_3)(\mb{k}_\theta(\mb{c}(\mb{x}_3)) - \mb{k}_T(\mb{x}_1)) \nonumber
\end{align}

We can now lower bound the first four lines of  \eqref{eq:pf_lipschitz_expansion} using Lipschitz constants. For example,
\begin{align}
L_\mb{f}h(\mb{x}_3) &- L_\mb{f}h(\mb{x}_1) + r_3 \mathfrak{L}_{L_{\mb{f}}h} \\
& \geq - \Vert L_\mb{f}h(\mb{x}_3) - L_\mb{f}h(\mb{x}_1) \Vert  + r_3 \mathfrak{L}_{L_{\mb{f}}h}\\
& \geq  \mathfrak{L}_{L_\mb{f}h}(r_3 - \Vert \mb{x}_1 - \mb{x}_2 + \mb{x}_2 - \mb{x}_3\Vert ) \\
& \geq \mathfrak{L}_{L_\mb{f}h}(r_3 - \Vert \mb{x}_1 - \mb{x}_2 \Vert - \Vert \mb{x}_2 - \mb{x}_3\Vert ) \geq 0 
\end{align}
since $\Vert \mb{x}_1 - \mb{x}_2  \Vert \leq r_1$ and $\Vert \mb{x}_2 - \mb{x}_3 \Vert\leq r_2$.

Applying these Lipschitz-based bounds for $L_\mb{f}h, L_\mb{g}h, \alpha \circ h, $ and $\varphi \Vert L_\mb{g}h\Vert^2  $ yields: 
\begin{align}
    \frac{d}{dt}h(\mb{x}_3, \mb{k}_\theta(\mb{c}(\mb{x}_3))) \geq  & -   \alpha(h(\mb{x}_3)) + \varphi \Vert L_\mb{g}h(\mb{x}_3)\Vert^2 \label{eq:pf_lipschitz_bound}\\
    &  + L_\mb{g}h(\mb{x}_3)(\mb{k}_\theta(\mb{c}(\mb{x}_3)) - \mb{k}_T(\mb{x}_1)) \nonumber
\end{align}
Additionally we can lower bound the final term using properties \eqref{eq:error_bound} and \eqref{eq:controller_lip_bound} of \etek ~as:
\begin{align}
     L_\mb{g}h(&\mb{x}_3)(\mb{k}_\theta(\mb{c}(\mb{x}_3)) - \mb{k}_T(\mb{x}_1)) \nonumber\\
     \geq  L_\mb{g}&h(\mb{x}_3)(\mb{k}_\theta(\mb{c}(\mb{x}_3)) - \mb{k}_\theta(\mb{c}(\mb{x}_1)) + \mb{k}_\theta(\mb{c}(\mb{x}_1)) - \mb{k}_T(\mb{x}_1)) \\
     \geq - \Vert &  L_\mb{g}h(\mb{x}_3) \Vert \bigg( \Vert \mb{k}_\theta(\mb{c}(\mb{x}_3)) - \mb{k}_\theta(\mb{c}(\mb{x}_1))\Vert \nonumber\\
     & \quad \quad \quad \quad \quad \quad \quad \quad + \Vert \mb{k}_\theta(\mb{c}(\mb{x}_1)) - \mb{k}_T(\mb{x}_1)\Vert \bigg)\\
     \geq -\Vert & L_\mb{g}h(\mb{x}_3) \Vert (\mathfrak{L}_{\mb{k}_\theta\circ \mb{c}} r_3 + M_e  ). \label{eq:pf_learning_bound}
\end{align}

Using \eqref{eq:pf_learning_bound} to lower-bound \eqref{eq:pf_lipschitz_bound} results in: 
\begin{align}
    \frac{d}{dt}&h(\mb{x}_3,  \mb{k}_\theta(\mb{c}(\mb{x}_3)))\\
    &\geq  -   \alpha (h(\mb{x}_3))  + \varphi\Vert L_\mb{g}h(\mb{x}_3) \Vert^2  \nonumber \\
    & \quad \quad \quad \quad \quad \quad \quad \quad - \Vert L_\mb{g}h(\mb{x}_3) \Vert( \mathfrak{L}_{\mb{k}_\theta\circ \mb{c}} r_3 + M_e),  \\
    &\geq -   \alpha(h(\mb{x}_3)) - \frac{1}{2\varphi}(\mathfrak{L}_{\mb{k}_\theta\circ \mb{c}}r_3 + M_e)^2, \label{eq:pf_final_bound}
\end{align}
where the final bound is achieved by completing the square and removing the positive term. 

To achieve forward invariance of $\mathcal{C}_\delta$ we note that $h(\mb{x}_3) = \alpha^{-1}\left(-\frac{1}{2\varphi}(\mathfrak{L}_{\mb{k}_\theta}r_3 + M_e)\right) \implies \frac{d}{dt}h(\mb{x}_3, \mb{k}_\theta(\mb{x}_3))\geq 0 $. Since the bound \eqref{eq:pf_final_bound} holds for all $\mb{x}_3 \in \partial  \mathcal{C} \oplus B_{r_2}$ and $\underline{\varphi}$ was chosen such that $\partial \mathcal{C}_\delta \subset \partial  \mathcal{C} \oplus B_{r_2} $ it is true that $\frac{d}{dt}h(\mb{x}_4, \mb{k}_\theta(\mb{x}_4)) = 0 $ for all $\mb{x}_4 \in \partial \mathcal{C}_\delta $. Thus by Nagumo's theorem \cite{nagumo1942lage} the set $\mathcal{C}_\delta $ is forward invariant and $\mathcal{C}$ is ISSf.
\end{proof}

To the best of our knowledge, this is the first result to establish a direct relationship between the safety of a system and the parameters of the imitation learning problem used to develop its controller. We recognize that finding and using the exact Lipschitz constants may be impractical, but note that due to their conservatism the CBF-compliant controller may be capable of achieving safety with far smaller values as demonstrated in simulation in Section \ref{sec:sim}.  


The learned controller $\mb{k}_\theta$ developed in Theorem \ref{thm:main} has mathematical guarantees of safety, but may result in behaviors significantly different than the expert controller in the interior of the safe set $\textrm{Int}(\mathcal{C})$ when the system is far from the boundary $\partial \mathcal{C}$. Therefore we present a corollary which may result in improved behavioral cloning on the interior of $\mathcal{C}$ due to increased sampling. The safety guarantees of this corollary follow immediately from Theorem \ref{thm:main}. 

\begin{corollary}
Let the dataset $\mathcal{D}$ satisfy the inequality 
\begin{align}
    \min_{\mb{x} \in \mathcal{D} } \Vert \mb{x}_1 - \mb{x}  \Vert  \leq r_1 , &\quad\quad\forall \mb{x}_1 \in  \mathcal{C} \label{eq:corollary_extra_sampling}
\end{align}
in place of \eqref{eq:thm_required_sampling} for some $r_1 > 0 $ . Let the remaining assumptions of Theorem \ref{thm:main} hold,
then the closed loop system \eqref{eq:dyn_ete} is ISSf with respect to $\mathcal{C}$ and safe with respect to $\mathcal{C}_\delta $ \eqref{eq:issf_safe_set}. 
\end{corollary}

\begin{proof}
$\mathcal{C}$ is a closed set so $\partial \mathcal{C} \subseteq \mathcal{C}$, thus \eqref{eq:corollary_extra_sampling} $\implies $ \eqref{eq:thm_required_sampling} and the conditions of Theorem \ref{thm:main} are met. 
\end{proof}

\begin{figure*}[t]
    \centering
    \includegraphics[width=\linewidth]{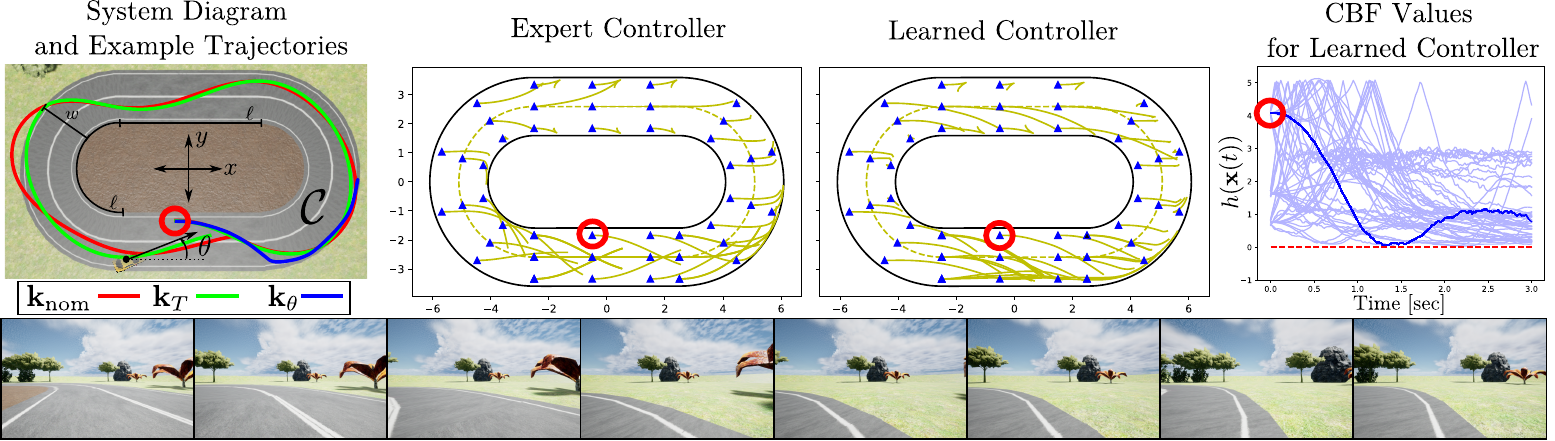}
    \caption{Results for the car. \textbf{Left: } Three 15 second long trajectories are shown starting from the same initial condition, the nominal controller generates the unsafe red  trajectory, the \ref{eq:trop} controller generates the safe green trajectory, and the learned controller generates the safe blue trajectory. \textbf{Center Left: } 3 second long trajectories starting at several initial conditions are shown for the expert controller $\mb{k}_T$. The blue triangles represent initial conditions and are all within the safe set. The yellow lines represent the trajectories of the car beginning at a blue triangle. \textbf{Center Right: } 3 second long trajectories for starting at several initial conditions shown for the learned controller $\mb{k}_\theta$. \textbf{Right: } CBF values $\min \{h_1(\mb{x}(t)), h_2(\mb{x}(t))\}$ achieved by the learned controller. Note all are greater than zero. The darker blue trajectory begins at the initial condition indicated by the red circles in each plot. \textbf{Bottom: } Images used by $\mb{k}_\theta $ for end-to-end control. From left to right, the first-person view of the trajectory indicated using the red circles starting at time $t= 0$ seconds and increasing by $0.25$ seconds. }
    \label{fig:car_results}
    \vspace{-0.5cm}
\end{figure*}

\section{Simulation Experiments} \label{sec:sim}

In this section we discuss the simulation results for safe vision-based end-to-end control of an inverted pendulum and a simplified car using \etek. 
In both cases the convolutional neural network used for end-to-end learning was MobileNetV2 \cite{sandler2018mobilenetv2} with an additional full-connected layer added to generate control inputs of the proper dimension. The full network has approximately 3.4 million parameters. Training was performed using the ADAM optimizer, an $\ell_2$ loss with $\ell_2$ weight decay, and batched training. The frequency of the observations was chosen to be $100 $ Hz and $60$ Hz for the pendulum and car respectively. The simulations were conducted using zero-order-hold control inputs of the same frequency with no latency. Additional information such as values of constants and image formatting can be found in our code \cite{codeCode}.

\subsection{Inverted Pendulum}


We first consider an inverted pendulum system with the states $\mb{x} = \lmat \theta & \dot{\theta}\rmat^\intercal $ with torque inputs $\tau \in \R$ as shown in Fig. \ref{fig:ip_results}. The dynamics and observation function of this system are given as: 
\begin{align}
    \dot{\mb{x}} & = \lmat \dot{\theta} \\ \sin\theta  \rmat + \lmat 0 \\ \tau \rmat,  &     \mb{y} = \mb{c}(\mb{x}) = \lmat \textrm{Img}(\mb{x})\\ \dot\theta \rmat
\end{align}
where $\textrm{Img}(\mb{x})$ represents the image of the system at state $\mb{x}$ as seen from a camera facing the inverted pendulum. A example images can be found in Figure \ref{fig:ip_results}. The learned controller is a function of the current image and velocity of the system, so an additional fully connected layer was added to incorporate the velocity into the end-to-end controller. 

The safe set for the inverted pendulum is chosen to be: 
\begin{align}
    h(\mb{x}) & = c - \mb{x}^\intercal \mb{P} \mb{x}, && \mathcal{C} = \{ \mb{x} \in \R^n ~|~ h(\mb{x}) \geq 0 \} 
\end{align}
where $\mb{P} \in \R^{2 \times 2}$ is such that $\mb{x}^\intercal\mb{P}\mb{x}$ is a control Lyapunov function derived from the continuous time algebraic Ricatti equation using feedback linearization and $c$ is chosen such that $\max_{\theta \in \mathcal{C}} |\theta |  = \pi/4$. This safe set is visualized in Fig. \ref{fig:ip_results}.

The expert controller is the \ref{eq:trop} controller with parameters $\varphi = 2 $, $\alpha (c)=  c$, and $a$ and $b$ chosen as the Lipschitz constants of $(L_\mb{f}h(\mb{x}) + \alpha ( h(\mb{x})) + \varphi \Vert L_\mb{g}h(\mb{x}) \Vert^2))$ and $L_\mb{g}h(\mb{x}) $ respectively over the compact set $\partial \mathcal{C}$ multiplied by the minimum sampling distance $r_1 = 0.01$. The nominal controller is $\mb{k}_\mathrm{nom}(\mb{x}) = -0.75 \theta$ which provides some torque to counteract gravity, but fails to stabilize the pendulum. The boundary of the safe set, $\partial \mathcal{C}$, is gridded and sampled uniformly with a minimum distance $r_1$ to create the training dataset $\mathcal{D}$.

\vspace{-0.15cm}

\subsection{Race Car}

Next we consider a simplified car given by the unicycle dynamics and observation function: 
\begin{align}
    \dot{\mb{x}}  = \lmat \cos \theta & 0 \\ \sin \theta & 0 \\ 0 & 1  \rmat \lmat v \\ \omega \rmat,  &&     \mb{y} = \mb{c}(\mb{x}) = \textrm{Img}(\mb{x})
\end{align}
where the state $\mb{x} = \lmat x & y & \theta \rmat^\intercal  $ is the planar position and heading angle and the input $\mb{u} = \lmat v & \omega \rmat^\intercal$ is the forward and angular velocities and  $\textrm{Img}(\mb{x}) $ represents the driver's first-person-view from the car at position $\mb{x}$. A series of example first-person-view images can be seen in Figure \ref{fig:car_results}. 

The safe set for the car is chosen to be the 0-superlevel set of the function $\min \{ h_1 , h_2 \} $ where:
\begin{align}
    h_i(\mb{x}) = 
    \delta \widehat{\mb{n}}^\intercal \widehat{\mb{d}}_i + \psi_i\cdot  \begin{cases}
    \rho_i ^2  - \left( \left(x - \frac{\ell}{2} \right)^2 + y^2 \right) , &  x\geq \frac{\ell}{2}\\
    \rho_i^2 - \left(\left(x + \frac{\ell}{2} \right)^2 + y^2 \right) , &  x\leq \frac{-\ell}{2} \\
   \rho_i ^2  - y^2, &  \mathrm{else} \nonumber 
    \end{cases} 
\end{align}
where $\rho_1 = (\ell/\pi + w) $, $\psi_1 = 1$, $\rho_2 = \ell/ \pi $,  and $\psi_2 = -1$. 
Additionally, $\delta = 0.1$, $\widehat{\mb{n}}$ is the unit vector in the car's heading direction, $\widehat{\mb{d}}_1 $ is the unit vector pointing perpendicularly inward from the outer boundary of the track through the car's position, $\widehat{\mb{d}}_2$ is the unit vector point perpendicularly outward from the inner boundary of the track through the car's position, $\ell$ is the length of the straight portions of the track, $w$ is the width of the track. An annotated diagram of the track can be found in Figure \ref{fig:car_results}. Given these functions, the safe set $\mathcal{C } = \{ \mb{x} \in \R^3 ~|~  \min \{ h_1(\mb{x}), h_2(\mb{x}) \}  \geq 0 \} $ is a subset of the track with an angle dependence where positions with heading angles pointed towards the center line are considered safer. 

The expert controller is the \ref{eq:trop} controller with the constraint simultaneously enforced for both $h_1$ and $h_2$ with parameters $\varphi = 0.5$, $a = 10^{-2}$, $b = 10^{-4}$, and  $\alpha(c) = 10 c $. $\partial \mathcal{C}$ was gridded and sampled uniformly with distance of $r_1=0.1$ to generate $\mathcal{D}$. 
We use Theorem \ref{thm:main} to guide the choice of these constants, but due to the difficulty of estimating the Lipschitz constants and the likely over-conservatism of the resulting controller we choose parameters which are likely much smaller than those required to sufficiently guarantee safety mathematically but we nonetheless succeed in demonstrating safety experimentally. 

The nominal controller used in the \ref{eq:trop} controller is: 
\begin{align}
    \mb{k}_\textrm{nom} & = \lmat K_p | r- r_\textrm{mid}| + F\\ K_r(r- r_\textrm{mid}) + K_\textrm{dir}(\widehat{\mb{n}}^\intercal \widehat{\mb{e}}_\textrm{mid} ) \rmat  \label{eq:k_nom}
\end{align}
where $K_p, F, K_r, K_\textrm{dit} \in \R_{>0}$, $r = \Vert [x, \; y ] \Vert $, $r_\textrm{mid} $ is the distance from the origin to the middle the track along a line passing through the car, $\widehat{\mb{e}}_\textrm{mid} $ is the unit vector from the vehicle to middle line of the track. This nominal controller is capable of circumnavigating the track, but is unsafe. 

\vspace{-0.1cm}

\subsection{Learning and Results}

For both the inverted pendulum and the car, the learned controller $\mb{k}_\theta$ is trained until convergence to minimize \eqref{eq:learning_opt_statement} where the \eqref{eq:trop} controller is the expert controller.

The safe set for both systems was gridded with initial conditions and simulated forward for 1 second for the inverted pendulum and 3 seconds for the car. For the car, $\theta = 0 $ was held constant for each initial condition and the interior of the track was sampled. The trajectories can be seen in Figures \ref{fig:ip_results} and \ref{fig:car_results}. For each trajectory the \ref{eq:trop} controller renders the system safe and this safety is transferred to the learned controller despite having different closed-loop behavior. The minimum of $h$ achieved for the inverted pendulum example was $~0.028$ and the smallest value of $\min \{ h_1(\mb{x}), h_2(\mb{x}) \} $ achieved by the car for all initial conditions was $~0.030$, indicating safety of both systems.

Even though the system deviates significantly from the expert trajectories, the learned controller successfully keeps the system inside of the safe set. Thus, although additional sampling can be performed to improve the learned behavior, sampling on the boundary of the safe set is sufficient to render it forward invariant.  



\bibliographystyle{IEEEtran}
\bibliography{cosner}

\end{document}